\documentclass[11pt, reqno]{amsart}

\usepackage{amsmath, amsthm, amscd, amsfonts, amssymb, graphicx, color}
\usepackage[bookmarksnumbered, colorlinks, plainpages]{hyperref}
\usepackage{dsfont}
\input{mathrsfs.sty}
\usepackage{subfigure}
 \makeatletter
\let\reftagform@=\tagform@
\def\tagform@#1{\maketag@@@{(\ignorespaces\textcolor{blue}{#1}\unskip\@@italiccorr)}}
\renewcommand{\eqref}[1]{\textup{\reftagform@{\ref{#1}}}}
\makeatother
\usepackage{hyperref}
\hypersetup{colorlinks=true, linkcolor=red, anchorcolor=green,
citecolor=cyan, urlcolor=red, filecolor=magenta, pdftoolbar=true}

\newtheorem{theorem}{Theorem}[section]
\newtheorem{lemma}[theorem]{Lemma}

\theoremstyle{definition}

\newtheorem{example}[theorem]{Example}

\theoremstyle{remark}

\numberwithin{equation}{section}
  
  \newcommand{\Mmn}{\mathbb{M}_m \to \mathbb{M}_n}
  \newcommand{\M}{ \mathbb{M}}
  \newcommand{\C}{ \mathbf{C}}
\newcommand{\tr}{ \mathrm{Tr}}
\newcommand{\ve}{ \mathrm{vec}}
\newcommand{\id}{ \mathrm{id}}
\newcommand{\SR}{ \mathrm{SR}}

\begin{document}
\setcounter{page}{1}

\title[Sharp Asymptotics for  $k$-Positivity Shifts under Singular Perturbations]{ Sharp Asymptotics for  $k$-Positivity Shifts under Singular Perturbations }
\author[M. Kian]{Mohsen Kian}

\address{Mohsen Kian: Department of Mathematics, University of Bojnord, P. O. Box 1339, Bojnord 94531, Iran}
\email{kian@ub.ac.ir }



\begin{abstract}
We show that, as the Choi matrix of a completely positive
perturbation direction approaches a singular limit, the shift
required for $k$-positivity can diverge more slowly than
that required for complete positivity.   For every $k\geq 2$, we construct a $(k-1)$-positive map $\Phi$ and a completely positive map $\Gamma$ on $\M_{k+1}$ with singular Choi matrix $\C_\Gamma$. Regularizing $\Gamma$ by $\Gamma_\varepsilon=\Gamma+\varepsilon\Delta$, where $\Delta(X)=\operatorname{Tr}(X)I$, we determine the least $t\geq 0$ for which $\Phi+t\Gamma_\varepsilon$ is $j$-positive. More precisely, $s_{j,\Gamma_\varepsilon}(\Phi)=0$ for $j<k$, while
$s_{k,\Gamma_\varepsilon}(\Phi)\sim c_k\varepsilon^{-(2k+1)/(2k+2)}$
and $s_{k+1,\Gamma_\varepsilon}(\Phi)=\varepsilon^{-1}$. The leading constant $c_k>0$ is determined explicitly. The fractional exponent is obtained through singular-value estimates and an explicit construction
of rank-$k$ perturbations of a nilpotent Jordan block.
\end{abstract}

\maketitle

%%%%%%%%%%%%%%%%%%%%%%%%%%%%%%%%%%%%%%%%%%%%%%%%%%%%%%%%%%%%%%%%%%%%
\section{Introduction}

 Let $\Phi:\Mmn$ be a Hermiticity-preserving linear map and write
$$\C_\Phi=\sum_{i,j=1}^m E_{ij}\otimes\Phi(E_{ij})\in\M_m(\M_n)$$
for its Choi matrix, where $E_{ij}$'s are standard matrix units for $\M_m$.   Jamio{\l}kowski's work \cite{Jam1972} and Choi's
characterization \cite{Cho1975} show that  $\Phi$ is completely positive if and only if $\C_\Phi\geq0$.

Between positivity and complete positivity lie the cones of $k$-positive maps. For
$1\leq k\leq\min\{m,n\}$, a map is $k$-positive when
$\id_k\otimes\Phi$ is positive, and  these cones become smaller as $k$ increases. Using the Choi--Jamio{\l}kowski theorem,  $k$-positivity of $\Phi$ corresponds to $k$-block positivity of   $\mathbf C_\Phi$.  Skowronek, St{\o}rmer, and \.{Z}yczkowski \cite{SSZ2009} formulate this correspondence as a generalized Jamio{\l}kowski--Choi theorem and study the associated dual cones of $k$-superpositive maps and $k$-separable operators. As a  direct    interpretation in entanglement theory,  Terhal and Horodecki \cite{TerHo2000} introduced the Schmidt number of a bipartite mixed state and       related  Schmidt-number classes to
the action of $k$-positive maps.

The structural physical approximation (SPA) mixes a positive
but not completely positive map with a noise channel to obtain
a completely positive map.    Korbicz, Almeida, Bae, Lewenstein, and Ac\'{i}n \cite{Kor2008}
   studied such approximations for a positive map
$\Lambda:\M_{d_A}\to\M_{d_B}$, considering the map
\begin{align}\label{eq_Ba}
  \widetilde{\Lambda}_p(X)=p\,\tr(X)\frac{I_{d_B}}{d_B}+(1-p)\Lambda(X),\qquad 0<p<1,
\end{align}
where $p$ is chosen as small as possible so that $\widetilde{\Lambda}_p$ is completely positive. In terms of Choi matrices, this is the problem of finding the
smallest $p$ such that
$$
\frac{p}{d_B}I_{d_A d_B}+(1-p)\mathbf C_\Lambda\ge0 .
$$

Augusiak et al. \cite{Aug2011} considered SPA constructions
in which the depolarizing map is replaced by other
entanglement-breaking channels.

Regula, Takagi, and Gu \cite{Reg2021} introduced
robustness-type quantities for Hermiticity-preserving maps
defined through optimization over admissible completely
positive auxiliary maps.
\subsection{Preliminaries}

Throughout this paper, $\M_d$ is the   algebra of all $d\times d$ matrices and $I_d$ denotes the identity matrix. A Hermitian matrix $A\in\M_d$ is called positive
semidefinite, denoted by $A\geq0$, if
$\langle x,Ax\rangle\geq0$ for every
$x\in\mathbb C^d$. It is called positive definite,
denoted by $A>0$, if $\langle x,Ax\rangle>0$
for every nonzero $x\in\mathbb C^d$.

For a matrix $A\in\M_d$, let
$\sigma_1(A)\geq\cdots\geq\sigma_d(A)\geq0$
denote its singular values, which are the eigenvalues of
$(A^*A)^{1/2}$. The smallest singular value is denoted by
$\sigma_{\min}(A)=\sigma_d(A)$, and satisfies  $\sigma_{\min}(A)=\min_{\|x\|=1}\|Ax\|$.

The Frobenius norm of $A=[a_{ij}]$ is defined by
$$
\|A\|_F
=\left(\tr(A^*A)\right)^{1/2}
=\left(\sum_{i,j=1}^d|a_{ij}|^2\right)^{1/2}
=\left(\sum_{i=1}^d\sigma_i(A)^2\right)^{1/2}.
$$
We also have
$|\det A|=\prod_{i=1}^d\sigma_i(A)$.
Under the usual vectorization of matrices,
$\|\operatorname{vec}(A)\|=\|A\|_F$. For standard facts concerning singular values and matrix norms, we refer to \cite{HJ2013}.

 A map $\Phi:\M_m\to\M_n$ is called positive if $\Phi(A)\geq0$ for all $A\geq0$. For any positive  integer $k$, the map $\Phi$ is called $k$-positive if the induced map $ \Phi_k=\id_k\otimes\Phi:\M_k(\M_m)\to\M_k(\M_n)$  defined by $\Phi_k([A_{ij}])=[\Phi(A_{ij})]$ is positive. If $\Phi$ is $k$-positive for any positive integer $k$, then it is called completely positive, see e.g. \cite{bh}.

Every non-zero vector  $v\in\mathbb{C}^m\otimes\mathbb{C}^n$ has a so-called Schmidt decomposition \cite[Theorem 9.1]{BeZ} as
$$v=\sum_{j=1}^r \lambda_j\,a_j\otimes b_j,$$
  where $\lambda_j>0$ and $\{a_j\}_{j=1}^r$ and $\{b_j\}_{j=1}^r$ are orthonormal families in $\mathbb{C}^m$ and $\mathbb{C}^n$, respectively.
The integer \(r\) is called the Schmidt rank of $v$, denoted by
$\SR(v)$. We set $\SR(0)=0$. Under the usual vectorization of matrices, we have
$$\SR(\ve(A))=\mathrm{rank}(A),\qquad A\in\M_d.$$

Assume that $\mathsf{S}_k=\{v\in\mathbb{C}^m\otimes\mathbb{C}^n : \SR(v)\leq k\}$ denotes the set of vectors with Schmidt rank at most $k$. A matrix $\mathbf{A}\in\M_m\otimes\M_n$ is said to be $k$-block positive if $\langle x,\mathbf{A}x\rangle \geq0$ for all  $x\in\mathsf{S}_k$.
The set of $k$-block positive matrices forms a closed convex
cone, denoted by $\mathsf{BP}_k$, given by
$$\mathsf{BP}_k=\left\{\mathbf{A}:\langle x,\mathbf{A}x\rangle\ge0\ \text{for all }x\text{ with }\SR(x)\le k\right\}.$$
 A map $\Phi$ is $k$-positive if and only if its Choi matrix $\C_\Phi$ is  $k$-block positive, see \cite{An04,SSZ2009}. The cone of $k$-block positive matrices corresponds to the cone of $k$-positive maps $\mathcal P_k$. The cones of $k$-positive maps form the nested hierarchy
\begin{align}\label{eq_incp}
 \mathcal P_d\subseteq\mathcal P_{d-1}\subseteq\cdots\subseteq\mathcal P_1,
\end{align}
where $\mathcal P_1$ is the cone of positive maps and $\mathcal P_d$ is the cone of completely positive maps,  see \cite{bh}.

\medskip
For a completely positive map $\Gamma$, define
$$s_{k,\Gamma}(\Phi)=\inf\{t\geq0:\Phi+t\Gamma\in\mathcal P_k\}.$$
We study this quantity when $\C_\Gamma$ is singular.  We regularize this direction by setting
$\Gamma_\varepsilon=\Gamma+\varepsilon\Delta$,
where $\Delta(X)=\mathrm{Tr}(X)I_d$ is the
completely depolarizing map.
For every $k\ge2$, we construct examples in which the shifts below level $k$ vanish, while the
$k$-positive shift satisfies
$$s_{k,\Gamma_\varepsilon}(\Phi)=c_k\varepsilon^{-\frac{2k+1}{2k+2}}(1+o(1)),$$
and the completely positive shift satisfies
$s_{k+1,\Gamma_\varepsilon}(\Phi)=\varepsilon^{-1}$.

\section{Main  Results}

For a Hermiticity-preserving map $\Phi:\mathbb M_d\to\mathbb M_d$, define the $k$-positivity shift as
$$s_k(\Phi)=\inf\{t\geq0:\Phi+t\Delta\text{ is }k\text{-positive}\},$$
where $ \Delta(X)=\mathrm{Tr}(X)I_d$.
The inclusion of the positivity cones \eqref{eq_incp} gives the hierarchy
$$s_1(\Phi)\leq s_2(\Phi)\leq\cdots\leq s_d(\Phi).$$

Since $\C_\Delta=I_{d^2}$, the $k$-positivity shift can be written as
\begin{align}\label{eq_var}
 s_k(\Phi)=\max\left\{0,-\min_{\substack{\|x\|=1\\
\SR(x)\leq k}}\langle x,\mathbf C_\Phi x\rangle
\right\}.
\end{align}

\begin{example}\label{ex_rs}
Consider the reduction map $\mathcal R_d(X)=\tr(X)I_d-X$ on $\M_d$.
Its Choi matrix is $\C_{\mathcal R_d}=I_{d^2}-\Omega\Omega^*$, where
$\Omega=\sum_{i=1}^{d}e_i\otimes e_i$. The variational characterization \eqref{eq_var} of the $k$-positivity shift
gives
$$s_k(\mathcal R_d)=\max\left\{0,\max_{\substack{\|x\|=1\\ \SR(x)\leq k}}
\bigl(|\langle\Omega,x\rangle|^2-1\bigr)\right\}.$$

Let $x$ be a unit vector with $\SR(x)\leq k$, and let
$X\in\M_d$ be defined by $x=\ve(X)$.
Then $\mathrm{rank}(X)\leq k$ and $\|X\|_F=1$.
Since $\Omega=\ve(I_d)$, we have $|\langle\Omega,x\rangle|=|\tr(X)|$.
Thus,
$$|\langle\Omega,x\rangle|\leq \sum_{i=1}^{k}\sigma_i(X)\leq \sqrt{k}\|X\|_F=\sqrt{k}.$$
Equality holds for $X=k^{-1/2}\sum_{i=1}^{k}E_{ii}$.
 Therefore,
    $s_k(\mathcal R_d)=k-1$ for every $ k=1,\dots,d$.

\end{example}

%%-------------------------------------------------------------
\medskip
%%---------------------------------------------------------------
For a completely positive map $\Gamma$, we define the directional $k$-positivity shift by
$$s_{k,\Gamma}(\Phi)=\inf\{t\geq0:\Phi+t\Gamma\text{ is }k\text{-positive}\}.$$

If the Choi matrix $\C_\Gamma$ is singular, the
directional shift may be infinite. Indeed, if there
exists $x\in\ker\C_\Gamma$ with $\SR(x)\leq k$ and
$\langle x,\C_\Phi x\rangle<0$, then
\[
\langle x,(\C_\Phi+t\C_\Gamma)x\rangle
=\langle x,\C_\Phi x\rangle<0
\]
for every $t\geq0$. Consequently,
$s_{k,\Gamma}(\Phi)=\infty$.

We consider the regularized family $\Gamma_\varepsilon=\Gamma+\varepsilon\Delta$ for $ \varepsilon>0.$
Since $\C_{\Gamma_\varepsilon}=\C_\Gamma+\varepsilon I_{d^2}>0,$
the regularized shifts are finite for every $\varepsilon>0$. We construct maps for which $s_{k,\Gamma_\varepsilon}(\Phi)$ grows at a fractional
power-law rate, whereas $s_{k+1,\Gamma_\varepsilon}(\Phi)=\varepsilon^{-1}$
as $\varepsilon\downarrow0$.

We will use the following singular value estimate.
\begin{lemma}\label{lm_sd}
Let $M\in\mathbb M_d$ and let $X\in\mathbb M_d$ be singular.
Then
\[
\sigma_{\min}(M)\leq \|M-X\|_F .
\]
\end{lemma}
\begin{proof}
Since $X$ is singular, there exists a unit vector $z$ such that
$Xz=0$. Hence
$$\|M-X\|_F\geq\|(M-X)z\|_2=\|Mz\|_2\geq\sigma_{\min}(M).$$
\end{proof}

\begin{theorem}\label{th_fd}
For every integer $k\geq2$, there exist maps $\Phi,\Gamma:\M_{k+1}\to\M_{k+1},$
where $\Phi$ is $(k-1)$-positive and $\Gamma$ is
completely positive with singular Choi matrix, such that
\begin{align}\label{eq_fh}
s_{j,\Gamma_\varepsilon}(\Phi)=
\begin{cases}
0, & 1\leq j<k,\\
c_k\varepsilon^{-\alpha_k}(1+o(1)), & j=k,\\
\varepsilon^{-1}, & j=k+1,
\end{cases}
\qquad (\varepsilon\downarrow0),
\end{align}
where
$$\alpha_k=\frac{2k+1}{2k+2},\qquad c_k=\frac{2k+1}{k+1}
\left(\frac{(k+1)^{k+1}}{(2k+1)k^k}\right)^{\frac{1}{2k+2}}.$$
The first and third identities in \eqref{eq_fh} hold for every
$\varepsilon>0$. Moreover,
$$s_{k,\Gamma_\varepsilon}(\Phi)\leq c_k\varepsilon^{-\alpha_k},\qquad \varepsilon>0.$$
\end{theorem}

\begin{proof}
Set $d=k+1$ and let $J\in\mathbb M_d$ denote the nilpotent Jordan block
$$J=\begin{pmatrix}
0&1&0&\cdots&0\\
0&0&1&\cdots&0\\
\vdots&&\ddots&\ddots&\vdots\\
0&\cdots&0&0&1\\
0&\cdots&0&0&0
\end{pmatrix},$$
and define
$$u=\frac{\ve(J)}{\sqrt{d-1}},\qquad v=\frac{\ve(I_d)}{\sqrt d}.$$
Since $\tr(J)=0$, the vectors $u$ and $v$ are orthonormal. Let
$Q=I_{d^2}-uu^*-vv^*$.

We define maps $\Phi,\Gamma:\M_d\to\M_d$ using their Choi matrices
$$\C_\Phi=-uv^*-vu^*+\tau Q,\qquad\text{and}\qquad \C_\Gamma=Q,$$
where $\tau>0$ will be chosen below. Clearly,
$\Gamma$ is completely positive and
$\mathbf C_\Gamma$ is singular.

We first show that $\Phi$ is $(d-2)$-positive. Under the vectorization correspondence,  a vector in $\mathrm{span}\{u,v\}$ corresponds to a matrix of the form
$X=aJ+bI_d$ with  $\det(aJ+bI_d)=b^d$. If $b\neq0$, then $X$ has full rank $d$, while if
$b=0$ and $a\neq0$, then $X=aJ$ has rank $d-1$. Therefore,
every nonzero vector in $\mathrm{span}\{u,v\}$ has Schmidt rank at least $d-1$. Thus,  the compact set
$$\mathcal S_{d-2}:=\{x\in\mathbb C^d\otimes\mathbb C^d:
\|x\|=1,\ \SR(x)\le d-2\}$$
 is disjoint from  $\ker Q$. Since $x\mapsto\langle x,Qx\rangle$ is continuous, the minimum $\delta=\min_{x\in\mathcal S_{d-2}}\langle x,Qx\rangle$ exists and is positive. We choose $\tau$ to satisfy   $\tau\geq \delta^{-1}-1$ so that
\begin{align}\label{eq_qx}
  (\tau+1)\langle x,Qx\rangle \geq \delta^{-1} \langle x,Qx\rangle \geq 1\qquad (x\in\mathcal S_{d-2}).
\end{align}
Furthermore, for every $x\in\mathcal S_{d-2}$, we have
\begin{align}\label{eq_qxc}
\begin{split}
\langle x,\C_\Phi x\rangle&=-2\mathrm{Re}(\langle x,u\rangle\overline{\langle x,v\rangle})
+\tau\langle x,Qx\rangle\\
&\geq-\left(|\langle x,u\rangle|^2+|\langle x,v\rangle|^2\right)+\tau\langle x,Qx\rangle.
\end{split}
\end{align}
 From definition of $Q$,  we have
\begin{align}\label{eq_qxc2}
|\langle x,u\rangle|^2+|\langle x,v\rangle|^2+\langle x,Qx\rangle=1.
\end{align}
Combining \eqref{eq_qxc} and \eqref{eq_qxc2} we obtain
$$\langle x,\C_\Phi x\rangle\geq-1+(\tau+1)\langle x,Qx\rangle\geq -1+\delta^{-1}\langle x,Qx\rangle\geq0 \qquad (x\in\mathcal S_{d-2}),$$
in which we use \eqref{eq_qx} in the last two inequalities.
 Therefore $\C_\Phi$ is
$(d-2)$-block positive, and hence $\Phi$ is
$(d-2)$-positive. Consequently, $s_{j,\Gamma_\varepsilon}(\Phi)=0$ for every $j=1,\dots,d-2$.

Consider  $\Gamma_\varepsilon=\Gamma+\varepsilon\Delta$ with
$\C_{\Gamma_\varepsilon}=Q+\varepsilon I_{d^2}$.  Using  the definition of $Q$, the Choi matrix of \(\Phi+t\Gamma_\varepsilon\) is
\begin{align}\label{eq_cph}
 \C_\Phi+t\mathbf C_{\Gamma_\varepsilon}
=-uv^*-vu^*+t\varepsilon(uu^*+vv^*)+(\tau+t(1+\varepsilon))Q.
\end{align}
Since $Q$ is the orthogonal projection onto
$\mathrm{span}\{u,v\}^\perp$, we have
\begin{align}\label{eq_decc}
\mathbb C^d\otimes\mathbb C^d=\mathrm{span}\{u,v\}\oplus\mathrm{Ran}Q.
\end{align}
 On $\mathrm{Ran}Q$, the   operator in \eqref{eq_cph} is
$(\tau+t(1+\varepsilon))I_{\mathrm{Ran}Q}$ which is positive for every $t\ge0$.
On the two-dimensional space $\mathrm{span}\{u,v\}$,  the matrix representation of operator in \eqref{eq_cph}   in the basis $\{u,v\}$ is
$$\begin{pmatrix}
t\varepsilon&-1\\
-1&t\varepsilon
\end{pmatrix}.$$
The eigenvalues of this matrix are $t\varepsilon\pm1$.  Hence positivity on this subspace is equivalent to $t\varepsilon\ge1$.  Therefore the smallest value of $t$ for which
$\C_\Phi+t\C_{\Gamma_\varepsilon}$ is positive semidefinite is
$s_{d,\Gamma_\varepsilon}(\Phi)=\varepsilon^{-1}$.

Since $\C_{\Gamma_\varepsilon}>0$,  the Choi characterization of $k$-positivity implies that
$\Phi+t\Gamma_\varepsilon$  is $(d-1)$-positive if and only if
$ \langle x,\C_\Phi x\rangle+t\langle x,\C_{\Gamma_\varepsilon}x\rangle\ge0$
for all unit vectors $x$ with $\SR(x)\le d-1$ or equivalentely
$$t\geq \frac{-\langle x,\C_\Phi x\rangle}{\langle x,\C_{\Gamma_\varepsilon}x\rangle}\qquad(\|x\|=1,\,\SR(x)\leq d-1).$$
Taking the maximum over all admissible unit vectors
and using $t\geq0$, we obtain
$$s_{d-1,\Gamma_\varepsilon}(\Phi)=\max\left\{ \frac{[-\langle x,\C_\Phi x\rangle]_+}{\langle x,\C_{\Gamma_\varepsilon}x\rangle}:\, \|x\|=1,\,\SR(x)\le d-1 \right\}.$$
According to \eqref{eq_decc} we can write a vector $x$ as
$x=\alpha u+\beta v+w$, while $\langle x,Qx\rangle=b=\|w\|^2$.

Using the definitions of $\C_\Phi$ and $\C_{\Gamma_\varepsilon}$ we obtain
\begin{align}
\begin{split}
\langle x,\C_\Phi x\rangle&=\langle x,(-uv^*-vu^*+\tau Q) x\rangle\\
&=-\overline{\alpha}\beta-\overline{\beta}\alpha
+\tau\langle x,Qx\rangle\\
&=-2\mathrm{Re}(\overline{\alpha}\beta)+\tau b,
\end{split}
\end{align}
and
$$\langle x,\C_{\Gamma_\varepsilon}x\rangle=\langle x,Qx\rangle+\varepsilon
=b+\varepsilon.$$
Therefore,
\begin{align}\label{eq_vaer}
s_{d-1,\Gamma_\varepsilon}(\Phi)
=\max_{\substack{\|x\|=1\\ \SR(x)\leq d-1}}
\frac{[2\mathrm{Re}(\overline{\alpha}\beta)-\tau b]_+}
{b+\varepsilon}.
\end{align}
Let $X$ be the matrix corresponding to $x$ under vectorization. Then
$\|X\|_F=1$ and $\mathrm{rank} X\le d-1$. Using $x=\alpha u+\beta v+w$, we have
$$X=\frac{\alpha}{\sqrt{d-1}}J+\frac{\beta}{\sqrt d}I_d+E=:M+E.$$
It follows from Lemma~\ref{lm_sd} that
\begin{align}\label{eq_sigm}
 \sigma_{\min}(M)\leq\|M-X\|_F=\|E\|_F=\sqrt b.
\end{align}
Let $\sigma_1(M)\ge\cdots\ge\sigma_d(M)$ be the singular values of $M$.  Using
$\det(aJ+bI_d)=b^d$, we obtain
\begin{align}\label{eq_smin}
  \sigma_{\min}(M)\prod_{j=1}^{d-1}\sigma_j(M)
   =\prod_{j=1}^{d}\sigma_j(M) =  |\det M|  =\left|\det\left(\frac{\alpha}{\sqrt{d-1}}J+\frac{\beta}{\sqrt d}I_d
\right)\right|=\frac{|\beta|^d}{d^{d/2}}
\end{align}
On the other hand, it follows from the arithmetic-geometric mean inequality that
\begin{align}\label{eq_sig}
\begin{split}
\prod_{j=1}^{d-1}\sigma_j(M) &\leq\left(\frac{1}{d-1}\sum_{j=1}^{d-1}\sigma_j^2(M)
\right)^\frac{d-1}{2}\\
&\leq
(d-1)^\frac{-(d-1)}{2}\left( \sum_{j=1}^{d}\sigma_j^2(M)\right)^\frac{d-1}{2}\\
&\leq (d-1)^\frac{-(d-1)}{2}\|M\|_F^{d-1}\\
&\leq(d-1)^\frac{-(d-1)}{2},
\end{split}
\end{align}
because $\|M\|_F^2=|\alpha|^2+|\beta|^2\leq1.$  Combining \eqref{eq_sigm},  \eqref{eq_smin} and \eqref{eq_sig} we get
\begin{align*}
 \sqrt b\geq \sigma_{\min}(M)\geq\frac{(d-1)^{(d-1)/2}}{d^{d/2}}|\beta|^d
\end{align*}
or
\begin{align}\label{eq_b}
 b\geq
\frac{(d-1)^{d-1}}{d^d}|\beta|^{2d}.
\end{align}
 Since $\tau b\ge0$ and $|\alpha|\le1$, we have
\begin{align}\label{eq_xc}
[2\mathrm{Re}(\overline{\alpha}\beta)-\tau b]_+ \leq 2|\beta|.
\end{align}
With  $z=|\beta|$, it follows from \eqref{eq_vaer}, \eqref{eq_b} and \eqref{eq_xc}  that
\begin{align}\label{eq_up}
s_{d-1,\Gamma_\varepsilon}(\Phi)
\leq\sup_{z\ge0}\frac{2z}{\frac{(d-1)^{d-1}}{d^d}z^{2d}+\varepsilon}.
\end{align}

Let
$$f(z)=\frac{2z}{\frac{(d-1)^{d-1}}{d^d}z^{2d}+\varepsilon}.$$
A direct differentiation gives $f'(z)=0$  when
$(2d-1)\frac{(d-1)^{d-1}}{d^d}z^{2d}=\varepsilon$. Hence the maximizing point satisfies
$z^{2d}=\frac{d^d\varepsilon}{(2d-1)(d-1)^{d-1}}$. Since $f(z)$ tends to zero as $z\to\infty$, this critical point gives the maximum. The maximum is equal to $c_k\varepsilon^{-\frac{2d-1}{2d}}$,
where
$$c_k=\frac{2d-1}{d}\left(\frac{d^d}{(2d-1)(d-1)^{d-1}}\right)^{\frac1{2d}}.$$
Substituting this maximum in \eqref{eq_up}, multiplying both sides by $\varepsilon^{\frac{2d-1}{2d}}$ and taking limit we obtain
\begin{align}\label{eq_lmsp}
\limsup_{\varepsilon\downarrow0}\varepsilon^{\frac{2d-1}{2d}}s_{d-1,\Gamma_\varepsilon}(\Phi)
\le c_k.
\end{align}

To obtain the matching lower bound, we construct
a family of matrices of rank $d-1$.
Let $h>0$ be a parameter and let $E_{d1}$ denote the matrix
unit with a single nonzero entry in position $(d,1)$. Define
$ X_h=J+hI_d+(-1)^dh^dE_{d1}$.  Expanding the determinant along the last row gives $\det X_h=0$.
Moreover, the upper-left $(d-1)\times(d-1)$ minor is
nonzero, so $\mathrm{rank}X_h=d-1$. Also,
$$
\|X_h\|_F^2=d-1+dh^2+h^{2d}.
$$
For the normalized vector $x_h= \frac{\ve(X_h)}{\|X_h\|_F}$, the coefficients in the orthogonal decomposition $x_h=\alpha_hu+\beta_hv+w_h$  are
$$\alpha_h=\frac{\sqrt{d-1}}{\|X_h\|_F},
\qquad\beta_h=\frac{h\sqrt d}{\|X_h\|_F},
\qquad b_h=\frac{h^{2d}}{\|X_h\|_F^2}.
$$
Therefore,   \eqref{eq_vaer}   gives
\begin{align}\label{wov}
s_{d-1,\Gamma_\varepsilon}(\Phi)\geq\frac{2h\sqrt{d(d-1)}-\tau h^{2d}
}{h^{2d}+\varepsilon(d-1+dh^2+h^{2d})}.
\end{align}
Set $h=q\varepsilon^{1/(2d)}$, where $q>0$.
Then \eqref{wov} yields
$$\varepsilon^{\frac{2d-1}{2d}}s_{d-1,\Gamma_\varepsilon}(\Phi)
\ge\frac{2q\sqrt{d(d-1)}-\tau q^{2d}\varepsilon^{1-\frac1{2d}}}{
q^{2d}+d-1+dq^2\varepsilon^{1/d}+q^{2d}\varepsilon}.$$
Taking limit gives
\begin{align}\label{eq_gq}
\liminf_{\varepsilon\downarrow0}\varepsilon^{\frac{2d-1}{2d}}
s_{d-1,\Gamma_\varepsilon}(\Phi)\ge\frac{2q\sqrt{d(d-1)}}{q^{2d}+d-1}.
\end{align}
Set
$g(q)=\frac{2q\sqrt{d(d-1)}}{q^{2d}+d-1}$.
A direct differentiation gives $g'(q)=0$ if and only if  $q^{2d}=\frac{d-1}{2d-1}$.
Since $g(q)\to0$ as $q\to0$ and $q\to\infty$, this critical point is the global maximum.
Substituting it yields
$$\max_{q>0}g(q)=\frac{2d-1}{d}\left(\frac{d^d}{(2d-1)(d-1)^{d-1}}\right)^{\frac1{2d}}=c_k.$$
Together with upper bound in \eqref{eq_lmsp}, this gives
$$\lim_{\varepsilon\downarrow0}\varepsilon^{\frac{2d-1}{2d}}s_{d-1,\Gamma_\varepsilon}(\Phi)
=c_k.$$
Since $d=k+1$, we obtain
$$c_k=\frac{2k+1}{k+1}\left(\frac{(k+1)^{k+1}}{(2k+1)k^k}\right)^{\frac1{2k+2}},$$
as required.
\end{proof}

%%%0000000000000000000000000000000000000000000000000000000000000000
We give the construction explicitly for $k=2$.
\begin{example}\label{ex_k2}
 Let $k=2$ and
$$J=\begin{pmatrix}
0&1&0\\
0&0&1\\
0&0&0
\end{pmatrix},\qquad u=\frac{\ve(J)}{\sqrt2},
\qquad v=\frac{\ve(I_3)}{\sqrt3},$$
and define $Q=I_9-uu^*-vv^*$. Consider the maps $\Phi,\Gamma:\M_3\to\M_3$
with Choi matrices $\C_\Phi=-uv^*-vu^*+5Q$ and  $\C_\Gamma=Q$.

The choice of the coefficient $5$ guarantees positivity of $\Phi$. Indeed, for every unit product vector $x$, we have  $|\langle u,x\rangle|^2\le\frac12$ and $|\langle v,x\rangle|^2\le\frac13$
so that
$$\langle x,Qx\rangle=1-|\langle u,x\rangle|^2-|\langle v,x\rangle|^2\ge\frac16\qquad\text{and}\qquad\langle x,\mathbf C_\Phi x\rangle
\ge-1+5\langle x,Qx\rangle+\langle x,Qx\rangle\ge0,$$
giving positivity of $\Phi$.  Moreover, $\Gamma$ is completely positive and has a singular Choi matrix.

For $\Gamma_\varepsilon=\Gamma+\varepsilon\Delta$, Theorem~\ref{th_fd} gives
$s_{1,\Gamma_\varepsilon}(\Phi)=0$ and $s_{2,\Gamma_\varepsilon}(\Phi)
=\frac53\left(\frac{27}{20}\right)^{1/6}\varepsilon^{-5/6}(1+o(1))$  and
$s_{3,\Gamma_\varepsilon}(\Phi)=\varepsilon^{-1}$.

\end{example}

%%%%%%%%%%%%%%%%%%%%%%%%%%%%%%%%%%%%%%%%%%%%%%%%%%%%%%%%%%%%%%%%%%%%

\end{document}